\documentclass[12pt]{article}
\usepackage{amsmath,amssymb,amsthm}
\usepackage[top=1in,bottom=1in,left=1in,right=1in]{geometry}
\usepackage{setspace}
\usepackage[linesnumbered,lined,ruled]{algorithm2e}
\usepackage{graphicx}
\usepackage{color}
\usepackage[fleqn]{mathtools}
\usepackage{multirow}
\usepackage[shortcuts]{extdash} 
\usepackage{authblk}
\usepackage{tikz}
\usepackage{xcolor}
\usepackage[font=small]{caption}
\usepackage{hyperref}

\newtheorem{definition}{Definition}[section]
\newtheorem{theorem}{Theorem}[section]
\newtheorem{lemma}{Lemma}[section]

\newcommand{\cdag}{\ensuremath{\mathtt{cdag}}}
\newcommand{\ctree}{\ensuremath{\mathtt{ctree}}}
\newcommand{\cforest}{\ensuremath{\mathtt{cforest}}}

\newcommand{\subg}[2]{\ensuremath{{#1}[#2]}}

\newcommand{\ance}[1]{\ensuremath{\mathcal{A}({#1})}}
\newcommand{\desc}[1]{\ensuremath{\mathcal{D}({#1})}}
\newcommand{\ancePlus}[1]{\ensuremath{\mathcal{A}^{+}({#1})}}
\newcommand{\descPlus}[1]{\ensuremath{\mathcal{D}^{+}({#1})}}
\newcommand{\anceR}[1]{\ensuremath{\mathcal{A}^{R}({#1})}}
\newcommand{\descR}[1]{\ensuremath{\mathcal{D}^{R}({#1})}}
\newcommand{\reverse}[1]{\ensuremath{#1^{R}}}

\newcommand{\Isa}{\ensuremath{\mathsf{is}}-\ensuremath{\mathsf{a}}}
\newcommand{\Partof}{\ensuremath{\mathsf{part}}-\ensuremath{\mathsf{of}}}
\newcommand{\Regulates}{\ensuremath{\mathsf{regulates}}}

\definecolor{myred}{rgb}{0.7686,0.1882,0.1686}
\definecolor{mygreen}{rgb}{0.1255,0.5020,0.3137}
\definecolor{myblue}{rgb}{0,0.3255,0.6235}
\definecolor{myyellow}{rgb}{0.9804,0.5843,0}

\newcommand\drawcircle[1]{%
  \begin{tikzpicture}                                                           
    \draw [fill={#1}] (0,0.8) circle [radius=0.1];                                       
  \end{tikzpicture}%
}

\date{\vspace{-5ex}} 

\begin{document}

\title{Enumerating consistent subgraphs of directed acyclic graphs: an insight into biomedical ontologies}

\newcommand\CoAuthorMark{\footnotemark[\arabic{footnote}]}

\author{Yisu Peng\footnote{Contributed equally to this work.}}
\author{Yuxiang Jiang\protect\CoAuthorMark}
\author{Predrag Radivojac\footnote{Corresponding author, email: \texttt{predrag@indiana.edu}}}
\affil{Department of Computer Science,\\ Indiana University, Bloomington, Indiana, U.S.A.}

\maketitle

\begin{abstract}
Modern problems of concept annotation associate an object of interest (gene, individual, text document) with a set of interrelated textual descriptors (functions, diseases, topics), often organized in concept hierarchies or ontologies. Most ontologies can be seen as directed acyclic graphs, where nodes represent concepts and edges represent relational ties between these concepts. Given an ontology graph, each object can only be annotated by a consistent subgraph; that is, a subgraph such that if an object is annotated by a particular concept, it must also be annotated by all other concepts that generalize it. Ontologies therefore provide a compact representation of a large space of possible consistent subgraphs; however, until now we have not been aware of a practical algorithm that can enumerate such annotation spaces for a given ontology. In this work we propose an algorithm for enumerating consistent subgraphs of directed acyclic graphs. The algorithm recursively partitions the graph into strictly smaller graphs until the resulting graph becomes a rooted tree (forest), for which a linear-time solution is computed. It then combines the tallies from graphs created in the recursion to obtain the final count. We prove the correctness of this algorithm and then apply it to characterize four major biomedical ontologies. We believe this work provides valuable insights into concept annotation spaces and predictability of ontological annotation.
\end{abstract}

\section{Introduction}

Ontologies have become a common means of concept annotation in computational biology and related fields \cite{Robinson2011}. A protein's molecular function \cite{Ashburner2000}, an effect of a genetic variant \cite{Vihinen2014}, or a patient's diagnosis \cite{Robinson2010} are typical examples in which biomedical entities such as macromolecules, mutations, or individuals are associated with sets of mutually dependent descriptors. The dependencies between these descriptors are often hierarchical, leading to the use of directed acyclic graphs as concept space representations.

A directed acyclic graph is a pair $(V,E)$, where $V$ is a set of vertices (nodes) and $E$ is a set of directed edges (links) between vertices such that no cycles can be formed. Each vertex in the graph is associated with a unique concept (term, description) and each edge is associated with a particular type of relational tie. For example, when annotating proteins as biomedical entities using the Gene Ontology graph \cite{Ashburner2000}, the terms ``nucleic acid binding'' and ``DNA binding'' are linked by edges of the type \Isa\ asserting that DNA binding is a more specific form of nucleic acid binding. Other types of relational ties include \Partof, \Regulates, and so on.

A typical biomedical entity is associated with a set of terms determined through experiment such as a molecular assay or a diagnostic procedure. A protein, for example, may be assigned terms ``DNA binding'' and ``RNA binding'', neither of which is a generalization of the other. To avoid annotation inconsistencies, this protein must also be annotated by the terms such as ``nucleic acid binding'' and all others that generalize either of the experimentally determined terms. More broadly, this implies that a biomedical object can only be annotated by a set of terms that respect the hierarchy -- a \emph{consistent subgraph} of the ontology. Unfortunately, (manual) experimental annotation is resource-demanding and often incomplete \cite{Poux2017}, giving rise to an entire field of computational prediction \cite{Radivojac2013, Jiang2016}. 

The development of computational prediction methods presents its own challenges. Although it can be performed by building a separate binary classifier for each concept in the ontology, this approach is currently competitive only for specialized ranking tasks; e.g., disease-gene prioritization \cite{Moreau2012}, since it does not exploit relationships between the terms. On the other hand, a more complete characterization is via learning structured outputs \cite{Sokolov2010} in which a method takes an object (e.g., a protein) and is asked to provide the totality of concepts with which this object might be associated (i.e., a consistent subgraph). However, the structured-output formulation generally falls under the extreme classification umbrella because the size of the output space is often exceedingly large. This poses problems in measuring similarity between annotations, evaluating accuracy of classification models, and optimization when solving the ``argmax problem'' \cite{Clark2013, Friedberg2017, Joachims2009}.

We identify now what we believe is an open problem in computational biology and computer science; that is, efficiently determining the exact number of consistent subgraphs in a given ontology. This problem has a linear-time solution for rooted trees \cite{Ruskey1981}, but to our knowledge no such algorithm exists for directed acyclic graphs. This paper therefore proposes a practical solution to this enumeration problem, proves its correctness, analyzes run-time complexity, and introduces various computational speedups. Using this new approach, we analyze four often-used ontologies from the biomedical domain and explore the space of possible annotations. We believe that the algorithms, software, and analysis carried out in this work will lead to better insights into concept annotation spaces and facilitate ontology quality assurance.

\section{A Motivating Example}
A growing number of concept annotation problems are formulated as the manual or computational assignment of a set of mutually related textual descriptors to some objects of interest. One of such problems is the computational prediction of protein function \cite{Friedberg2017}, which can be broadly operationalized as follows: 

\vspace{2mm}

\begin{quote} 
\emph{Given}: (1) an amino acid sequence with auxiliary data such as structure, expression, interactions, etc.~of a protein $p$ with unknown or incomplete function; (2) training data that includes sequences, structures, or systems data corresponding to a (large) set of proteins, some of which have their true biological functions available; (3) a Gene Ontology (GO); i.e., a concept hierarchy used to represent biological functions of proteins in a structured and easy-to-compute-on form.

\vspace{1mm}

\emph{Objective}: provide a set of GO terms that are most likely to be the true (experimental) annotation of $p$. 
\end{quote} 

\vspace{2mm}

\noindent The objects of interest here are proteins and the set of textual descriptors of protein function is given by GO -- an ontology with a directed acyclic graph structure where each node represents a textual descriptor and each edge represents a particular type of a relational tie between two descriptors \cite{Ashburner2000}. 

\begin{figure}
    \centering
    \includegraphics{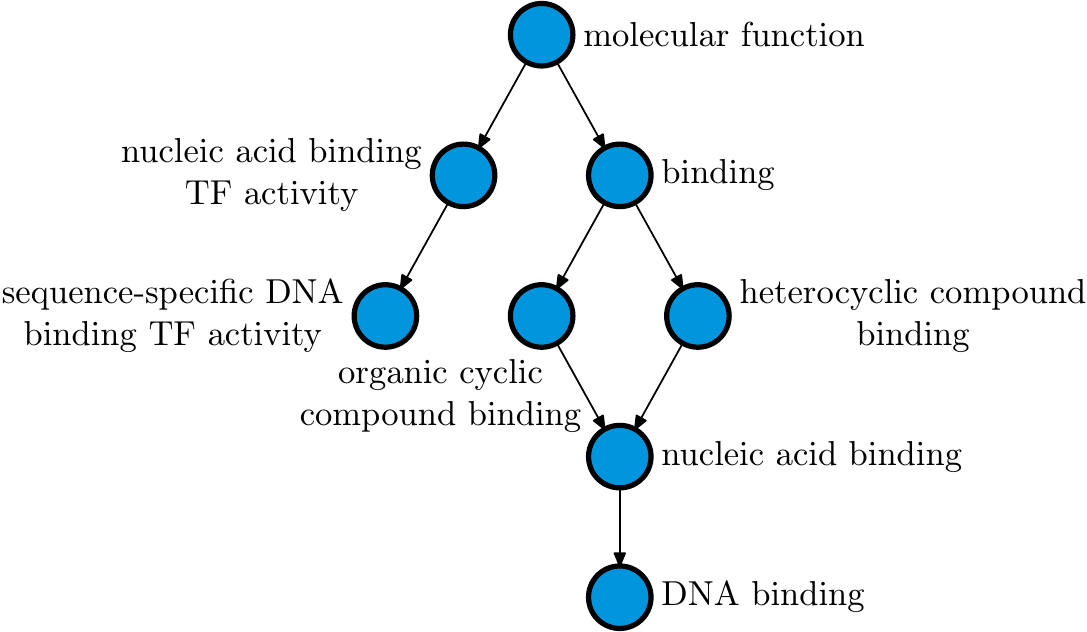}
    \caption{The functional annotation of the friend leukemia integration 1 transcription factor isoform 1 [\emph{FLI1}; \textit{Homo sapiens}] (RefSeq ID: NP\_002008.2) as a consistent subgraph of the molecular function ontology. The arrows in this graph indicate an \textsf{is-a} relationship and are drawn in the reverse direction.}
    \label{fig:annotation}
\end{figure}

An example of such an annotation is shown in Figure~\ref{fig:annotation}, where 8 terms from the molecular function domain have been assigned to this protein. Due to the hierarchical organization of GO, both the set of experimentally determined terms and the set of computationally predicted terms must respect this hierarchy. As shown in this example, the annotation of the term ``DNA binding'', implies the annotation of all the other GO terms that conceptually generalize it; e.g., ``nucleic acid binding'', ``binding'', etc. Typically, the ontology used to represent the annotation space of proteins contains thousands to tens of thousands of terms, whereas the true annotation of a protein consists of tens to at most hundreds of terms. Because the task of a prediction algorithm is to find the most likely annotation, it must devise an efficient procedure to search through the space of \emph{all} possible annotations.

Most biomedical ontologies have grown over the years to contain a large number of terms. Computationally selecting one such ``winning'' annotation; i.e., a set of terms, or even providing a short list of most likely annotations, is a significant challenge \cite{Joachims2009, Sokolov2010}. This prediction problem thus belongs to a so-called extreme classification scenario because the number of possible (discrete) annotations the algorithm must consider is astronomically large. In fact, we noticed that it is not even possible to give an exact number of possible annotations for a protein. Therefore, an answer to such a simple question (``What is the the number of possible GO annotations a protein can be assigned?'') requires the development of a practical counting algorithm. The resulting counts can, in turn, give insight into the nature and the difficulty of the computational function annotation of biological macromolecules.\footnote{Reasonable approximations can be provided by calculating the lower and upper bounds, as we have done later in Section \ref{sec:results}. Neither of those, however, provides a full intellectual satisfaction when an exact count can actually be computed.}

It is important to mention that the annotation of biological macromolecules is one of the most interesting examples of concept annotation, primarily because of its biomedical significance but also because of the sizes of the available ontologies. Similar situations, however, arise beyond computational biology, as in the fields of text mining \cite{Grosshans2014} and computer vision \cite{Movshovitz2015}. 

\section{Preliminaries}

\subsection{Basic Concepts and Notation}
Let $\mathcal{G} = (V, E)$ be a \textit{directed} graph, where $V$ is a set of vertices representing concepts and $E\subseteq V\times V$ is a collection of ordered pairs $(u, v)$ representing directional relationships, $u \to v$, between two concepts. A sequence of vertices $u_{1}, u_{2}, \dots, u_{k}$ is called a \textit{walk} if $(u_{i}, u_{i+1})\in E$ for $i = 1, 2, \dots, k-1$. A walk of distinct vertices except for the identical starting and ending vertices is called a \textit{cycle}. A directed graph that does not contain cycles is referred to as \textit{directed acyclic graph}~(DAG).

Given two vertices $u, v \in V$ in a DAG, $u$ is said to be an ancestor of $v$ and $v$ is said to be a descendant of $u$ if there exists a walk from $u$ to $v$. We denote a set of all ancestors of $v$ as $\ance{v}$ and a set of all descendants of $u$ as $\desc{u}$. We next define $\ancePlus{v} = \{v\} \cup \ance{v}$ as the set of extended ancestors of $v$ and $\descPlus{u} = \{u\} \cup \desc{u}$ as the set of extended descendants of $u$. Finally, if $(u, v) \in E$, the vertex $u$ is said to be a parent of $v$, whereas $v$ is said to be a child of $u$. We denote the set of all parents of $v$ as $\mathcal{P}(v)$ and the set of all children of $u$ as $\mathcal{C}(u)$. 

\subsection{Transitivity of Relational Ties}
When an object is annotated with ontological concepts, it is often considered that all ancestors of those annotated concepts should be automatically assigned to the object. For example, annotating the function of a protein with ``enzyme binding'' also implicitly annotates it with ``protein binding'', ``binding'' and, finally, the root term ``molecular function''. This type of reasoning requires all involved relationships between concepts to be transitive.

Biomedical ontologies, however, usually contain various types of relationships between concepts, some of which are not transitive. Therefore, we only consider \Isa\ and \Partof\ relationships, both of which maintain transitivity and permit reasoning about ancestral concepts. It is also worth noting that we define the direction of edges to be pointing from the general terms to specific so that the depth of a node aligns with the increasing resolution of the descriptors. We show in Section~\ref{sec:reversal} that the directionality of edges has no impact on the total count. Throughout this work, we consider an ontology $\mathcal{O} = (V, E)$ to be a DAG, where edges represent transitive relationships.

\subsection{Consistent Subgraphs}
Let $\mathcal{O}=(V, E)$ be an ontology and $S \subseteq V$ a set of vertices. A subgraph $(S, E^S)$ is said to be induced from the original graph $\mathcal{O}$ by $S$ if $E^S$ is the largest subset of pairs $(u, v)$ from $E$ such that both $u, v \in S$. We denote such vertex-induced subgraph as $\subg{\mathcal{O}}{S}$. We also use $\mathcal{O}[-S] = (V-S, E^{V-S})$ to denote the subgraph induced by vertices other than $S$.

\begin{definition}
A subgraph $\subg{\mathcal{O}}{S} = (S, E^S)$ with respect to the original graph $\mathcal{O} = (V, E)$ is called \emph{consistent} if $\forall v \in S$, $(u, v) \in E \implies u \in S$.
\end{definition}

\section{Basic Algorithms}

\subsection{Problem Specification}
Given an ontology $\mathcal{O} = (V, E)$, our goal is to develop a practical algorithm that enumerates all consistent subgraphs of $\mathcal{O}$. We allow the graph to have more than a single root (a vertex with no incoming edges) as well as to be disconnected. 

An example of the enumeration problem is shown in Figure~\ref{fig:example}. We generally observe that the number of consistent subgraphs is bounded from below by $2^{\ell}$, where $\ell$ is the total number of leaf vertices (those with no outgoing edges), and from above by $2^{|V|}$. The structure of the graph, however, determines the exact count and its proximity to either of the bounds. If the input graph is a chain of $|V|$ vertices ($\ell=1$), the total number of consistent subgraphs equals $|V|+1$. On the other hand, if the original graph is a set of $|V|$ disconnected vertices ($\ell=|V|$), there are $2^{|V|}=2^{\ell}$ consistent subgraphs. This analysis suggests that enumerating consistent subgraphs has a straightforward intractable solution of listing all $2^{|V|}$ vertex-induced subgraphs of the ontology and checking for the consistency of each such subgraph.

\begin{figure}
\centering
\includegraphics[width=\columnwidth]{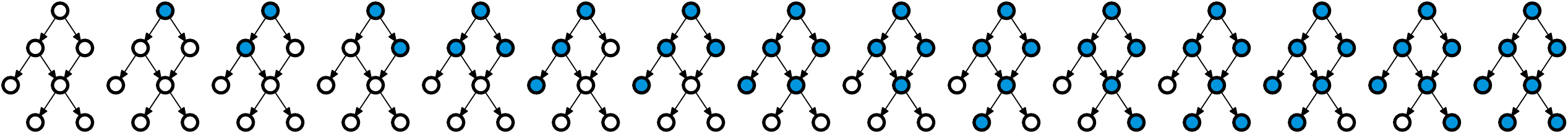}
\caption{Consistent subgraphs of an ontology $\mathcal{O}=(V,E)$ with $|V|=7$ vertices and $|E|=7$ edges. 
Observe that the reversal of all edges in the graph would lead to a reversed graph with the same number of consistent subgraphs (white vertices; Theorem \ref{thm:reverse}).}
\label{fig:example}
\end{figure}

We use $\cdag(\mathcal{O})$ to denote the desired function that takes a directed acyclic graph $\mathcal{O}$ as input and returns the number of consistent subgraphs in that graph. We use $\ctree(\mathcal{T})$ and $\cforest(\mathcal{F})$ for the special cases where the input graph is a rooted tree $\mathcal{T}$ or a forest $\mathcal{F}$, respectively.

\subsection{Counting Consistent Subgraphs in Trees}

We first discuss a special case where the input graph is a rooted tree; that is, when each non-root vertex has a single parent. In this case, there exists a linear algorithm in the number of vertices; see Lemma 1 in \cite{Ruskey1981}. We provide this solution in Algorithm \ref{algo:tree} with a minor modification resulting from the fact that our algorithm includes an empty tree in the total count. This algorithm naturally extends to collections of rooted trees. One can enumerate subtrees for each tree and take the product as the total count. We refer to this extended algorithm as $\cforest$~(not shown).

\begin{algorithm}[]\footnotesize
\SetKwInOut{Input}{Input}
\SetKwInOut{Output}{Output}
\SetKwBlock{Begin}{begin}{end}
\Input{A tree $\mathcal{T}_r$, rooted at $r$.}
\Output{The number of consistent subgraphs in $\mathcal{T}_r$.}
\SetAlgoLined
\SetKwProg{Fn}{Function}{}{end}
\Fn{$\ctree(\mathcal{T}_r)$}{
    \uIf{$\mathcal{T}_r$ is empty}
    {
        \Return 1\; \label{algo:tree:empty}
    }
    \Else
    {
        \Return $1+\prod_{u \in \mathcal{C}(r)} \ctree(\mathcal{T}_u)$\; \label{algo:tree:recursion}
    }
}
\caption{Counting the number of consistent subgraphs in rooted trees \cite{Ruskey1981}.}
\label{algo:tree}
\end{algorithm}

Algorithm~\ref{algo:tree} recursively traverses a tree in a pre-order manner. For any subtree rooted at vertex $v$, the number of consistent subtrees that contain $v$  equals the product of all subcounts from its subtrees rooted at each child. Additionally, we add $1$ for the only consistent subtree that does not contain $v$; i.e., the empty tree. The recursion terminates at the empty tree whose count is one. 

\begin{algorithm}[]\footnotesize
\SetKwInOut{Input}{Input}
\SetKwInOut{Output}{Output}
\SetKwBlock{Begin}{begin}{end}
\Input{A directed acyclic graph $\mathcal{O}$.}
\Output{The number of consistent subgraphs in $\mathcal{O}$.}
\SetAlgoLined
\SetKwProg{Fn}{Function}{}{end}
\Fn{$\cdag(\mathcal{O})$}{
    \uIf{$\mathcal{O}$ is a forest}
    {
        \Return $\cforest(\mathcal{O})$\;
    }
    \Else
    {
        Pick any vertex $u$ as the pivot\;
        \Return $\cdag(\subg{\mathcal{O}}{-\descPlus{u}}) + \cdag(\subg{\mathcal{O}}{-\ancePlus{u}})$\; \label{algo:dag:eqn}
    }
}
\caption{Counting the number of consistent subgraphs in directed acyclic graphs.}
\label{algo:dag}
\end{algorithm}

\subsection{Counting Consistent Subgraphs in Directed Acyclic Graphs}

Directed acyclic graphs generalize trees in that they allow for multi-parent vertices. 
Such vertices, however, break Algorithm~\ref{algo:tree} because the recursive branches are no longer independent. Algorithm~\ref{algo:dag} circumvents this problem by recursively decomposing a graph into two strictly smaller subgraphs according to a selected \emph{pivot} vertex. We will show in the next section that the number of consistent subgraphs in the two smaller graphs add up to be the number for the original graph (Line~\ref{algo:dag:eqn}, Algorithm~\ref{algo:dag}). The algorithm continues recursive enumeration until the graph becomes a forest, in which case it calls $\cforest$. Figure~\ref{fig:pivot} illustrates the process of graph decomposition with respect to the pivot vertex $u$. We note that any vertex can serve as pivot and will discuss the selection of pivots and how they impact the run time in Sections \ref{sec:pivotselection} and \ref{sec:runtime}.

\begin{figure*}[ht]
    \centering
    \includegraphics[width=0.8\textwidth]{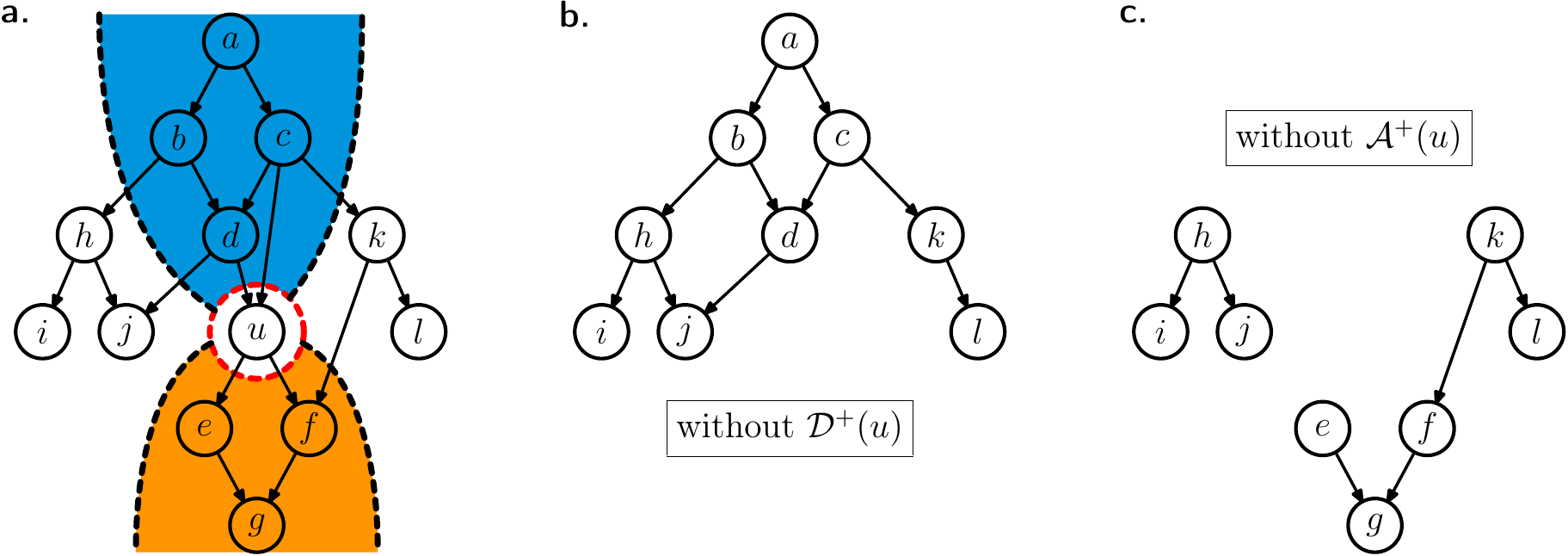}
    \caption{Illustration of graph decomposition. The enumeration problem of the original graph from panel (a) is split into two subproblems based on the pivot vertex $u$; shown in panels (b) and (c). The count in (b) corresponds to the number of consistent subgraphs in (a) that do not include $u$, while the count in (c) corresponds to the count of consistent subgraphs in (a) that include $u$. In panel (a), the set of descendants of $u$ is shaded in orange and the set of ancestors is shaded in blue.}
    \label{fig:pivot}
\end{figure*}

\subsection{Correctness and Complexity of the Algorithm}
\label{sec:correctness}
We first observe that the size of the problem in the number of vertices is guaranteed to decrease during recursive calls, thus ensuring that the algorithm terminates after a finite number of iterations. Next, we justify the equation corresponding to the Line~\ref{algo:dag:eqn} in Algorithm~\ref{algo:dag},
\begin{align}
\cdag(\mathcal{O}) = \cdag\left(\subg{\mathcal{O}}{-\descPlus{u}}\right) + \cdag\left(\subg{\mathcal{O}}{-\ancePlus{u}}\right). \label{eqn:decompose}
\end{align}

\begin{lemma} \label{lemma:withoutD}
Let $\cdag(\mathcal{O}\vert \neg u)$ be the number of consistent subgraphs in $\mathcal{O}$ that do not contain $u$. We have
$\cdag(\mathcal{O}\vert\neg u) = \cdag(\subg{\mathcal{O}}{-\descPlus{u}})$.
\end{lemma}

\begin{proof}
The equal cardinality of the two sets of consistent subgraphs is demonstrated by showing that both sets are contained in each other. For any $S \subseteq V-\descPlus{u}$ that induces a consistent subgraph of $\subg{\mathcal{O}}{-\descPlus{u}}$, it also induces a unique consistent subgraph in $\mathcal{O}$. Also, since none of them contains $u$, we have $\cdag(\subg{\mathcal{O}}{-\descPlus{u}}) \leq \cdag(\mathcal{O}\vert \neg u)$. Conversely, for any consistent subgraph induced by $S$ such that $u \notin S$, we have $\forall v \in \descPlus{u}, v \notin S$ by the definition of consistency. Therefore, $S$ also induces a consistent subgraph in $\subg{\mathcal{O}}{-\descPlus{u}}$. That is, $\cdag(\mathcal{O}\vert \neg u) \leq \cdag(\subg{\mathcal{O}}{-\descPlus{u}})$.

\end{proof}

\begin{lemma} \label{lemma:withoutA}
Let $\cdag(\mathcal{O}\vert u)$ be the number of consistent subgraphs in $\mathcal{O}$ that contain $u$. We have
$\cdag(\mathcal{O}\vert u) = \cdag(\subg{\mathcal{O}}{-\ancePlus{u}})$.
\end{lemma}

\begin{proof}
As in Lemma~\ref{lemma:withoutD}, for any $S \subseteq V-\ancePlus{u}$ that induces a consistent subgraph of $\subg{\mathcal{O}}{-\ancePlus{u}}$, $S \cup \ancePlus{u}$ also induces a unique consistent subgraph (that contains $u$) in the original graph. That is, $\cdag(\subg{\mathcal{O}}{-\ancePlus{u}}) \leq \cdag(\mathcal{O}\vert u)$. Also, for any consistent subgraph induced by $S$ and $u \in S$, we have $\ancePlus{u} \subseteq S$ by the definition of consistency. Note that the uniqueness of $S$ implies the uniqueness of $S-\ancePlus{u}$. We can see that the subgraph induced by $S-\ancePlus{u}$ in $\subg{\mathcal{O}}{-\ancePlus{u}}$ is consistent. Given $\forall w \in S-\ancePlus{u}$, and $(v, w)$ being an edge in $\subg{\mathcal{O}}{-\ancePlus{u}}$, we must have $v \in S-\ancePlus{u}$ as well, due to the consistency of $\subg{\mathcal{O}}{S}$ with respect to the original graph. That is, $\cdag(\mathcal{O}\vert u) \leq \cdag(\subg{\mathcal{O}}{-\ancePlus{u}})$.

\end{proof}

\begin{theorem}
Given an ontology $\mathcal{O}=(V,E)$ and any $u \in V$, the number of consistent subgraphs in $\mathcal{O}$ equals the sum of the numbers of consistent subgraphs in $\subg{\mathcal{O}}{-\descPlus{u}}$ and $\subg{\mathcal{O}}{-\ancePlus{u}}$.
\end{theorem}

\begin{proof}
Equation~\ref{eqn:decompose} holds by combining Lemmas~\ref{lemma:withoutD} and \ref{lemma:withoutA}.
\end{proof}

To analyze complexity of the algorithm, let $n$ be the number of vertices in the graph and $m$ be the number of multi-parent vertices. Assuming a multi-parent vertex is always selected as pivot, we can express the run time complexity $T(n)$ via the following recurrence
\begin{equation*}
T(n) \leq T(n-1) + T(n-3) + f(n), \label{eqn:worst_bound}
\end{equation*}
where $f(n)$ incorporates the time to select the pivot, split the graph and add two large integers. Let us further assume that the larger of the two graphs after decomposition contains $n - n/k$ elements, where $2 \leq k \leq n$. It is now straightforward to show that $T(n) = O(f(n) 2^{\min(m, s(k))})$, where $s(k)=O(n)$ if $k=O(n)$ and $s(k)=O(\log n)$ if $k=O(1)$.

We can now see that the algorithm is exponential in the worst case; however, it reduces to a polynomial algorithm when $m=O(\log n)$ or when $k=O(1)$. Assuming linear time to conduct graph decomposition and a constant time for addition/multiplication, we obtain $T(n)=O(n^2)$. 

\section{Advanced Algorithms}
The run-time of the algorithm heavily depends on the structure of the ontology and the selection of pivots. Here we discuss several practical considerations aimed at accelerating Algorithm \ref{algo:dag}. Once we conclude this discussion, the full method will be presented in  Algorithm \ref{algo:fulldag} (Section \ref{sec:results}).

\subsection{Pruning Branching Components}\label{sec:pruning}
It is easy to observe that when the ontology consists of multiple connected components, these components can be independently and, if needed, simultaneously processed. We take this reasoning a step further to consider a special scenario of nearly disconnected graphs where (i) the two components are connected via a single vertex and (ii) all vertices in one component are descendants of this vertex.

\begin{definition}
Given a graph $\mathcal{O} = (V,E)$ and $u \in V$, $\subg{\mathcal{O}}{\desc{u}}$ is called a \emph{branching component} if \,$\forall v \in V-\descPlus{u}$ and $ \forall w \in \desc{u}, (v, w) \notin E$. Vertex $u$ is called a \emph{branching vertex}.
\end{definition}

\begin{figure*}[ht]
    \centering
    \includegraphics[width=\textwidth]{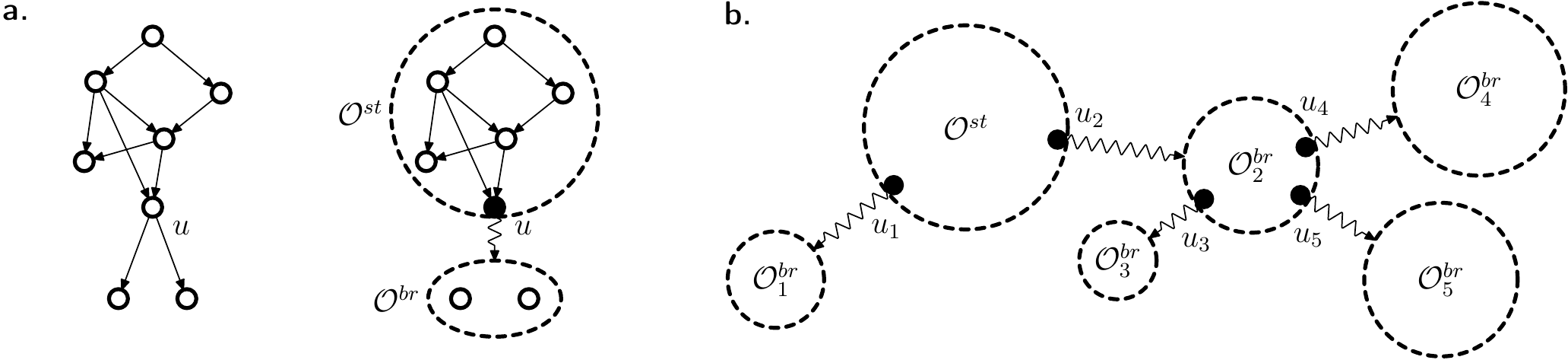}
    \caption{Illustration of branching components. Panel~(a) shows a branching vertex $u$ that separates the graph into a stem component $\mathcal{O}^{st}$ and a branching component $\mathcal{O}^{br}$. The collection of edges from $u$ to $\mathcal{O}^{br}$ is replaced by a zigzag arrow. Panel~(b) shows a component-wise tree structure.}
    \label{fig:prune}
\end{figure*}

Figure~\ref{fig:prune}a gives an example in which $u$ is a branching vertex, since the removal of $u$ disconnects $\desc{u}$ (i.e., the branching component, $\mathcal{O}^{br}$) from the rest of the graph. We refer to the remaining part of the graph as the stem component, $\mathcal{O}^{st}$. More generally, Figure~\ref{fig:prune}b shows a graph with a component-wise tree structure, where branching vertices serve as hinges of branching component to their corresponding stems. We will use $(\mathcal{O}^{st}, \mathcal{O}^{br}, u)$ to denote the desired structure.

Given $(\mathcal{O}^{st}, \mathcal{O}^{br}, u)$, we demonstrate that $\cdag({\mathcal{O}})$ can be decoupled into two sequential subproblems: (i) $\cdag(\mathcal{O}^{br})$ and (ii) $\cdag(\mathcal{O}^{st})$. We use $\varphi(u)$ for the subtotal of consistent subgraphs in the branching component $\mathcal{O}^{br}$. We also notice that the entire branching component can be pruned once $\varphi(u)$ is computed, making $u$ a leaf vertex in $\mathcal{O}^{st}$. Therefore, we modify the algorithm so as to allow a subtotal count $\varphi(u)$ for every vertex as if a branching component has been pruned from $u$. Notice that $\varphi(u) = 1$ for all intermediate vertices and original leaves.

With the introduction of $\varphi(u)$, the recursive equation in Algorithm~\ref{algo:tree} becomes
\begin{align} \label{eqn:tree_recursion_gen}
    \ctree(\mathcal{T}_r) = \varphi(r) + \prod_{u\in\mathcal{C}(r)}\ctree(\mathcal{T}_u).
\end{align}

\noindent Similarly, Equation~\ref{eqn:decompose}; i.e., Line~\ref{algo:dag:eqn} in Algorithm~\ref{algo:dag}, must be modified to 
\begin{align} \label{eqn:decompose_gen}
    \cdag(\mathcal{O}^{st}) = \cdag\left(\subg{\mathcal{O}^{st}}{-\descPlus{u}}\right) + \varphi(u)\cdot \cdag\left(\subg{\mathcal{O}^{st}}{-\ancePlus{u}}\right),
\end{align}
\noindent where $\varphi(u)$ accounts for the fact that for any consistent subgraph $S_i$ in the pruned $\mathcal{O}^{br}$ and any consistent subgraph $S_j$ in $\subg{\mathcal{O}^{st}}{-\ancePlus{u}}$, $\subg{\mathcal{O}}{S_i\cup S_j \cup \ancePlus{u}}$ is a distinct consistent subgraph in $\mathcal{O}$. The approach naturally extends to multiple (hierachical) branching components such that we compute the subtotal of consistent subgraphs within each component and agglomerate them in a reversed topological order.

The pruning operation is preferred before each instance of decomposition for two main reasons: (i) it divides the problem into smaller non-overlapping subproblems, while a direct decomposition usually results in substantial overlapping subproblems; (ii) although a full parallelization over components is restricted since stem components have to be computed only after all of their branching components are finished, the unordered components can be computed simultaneously. For example, as in Figure~\ref{fig:prune}b, $\mathcal{O}^{br}_1$, $\mathcal{O}^{br}_3$, $\mathcal{O}^{br}_4$ and $\mathcal{O}^{br}_5$ can be computed in parallel.

\subsection{Reverse Graphs}\label{sec:reversal}
Let $\reverse{\mathcal{O}}=(V,E^R)$ be the reverse graph of $\mathcal{O}$, where $E^R = \{ (u, v) \vert (v, u) \in E \}$. We show that the number of consistent subgraphs in $\mathcal{O}$ equals that in $\reverse{\mathcal{O}}$.
\begin{lemma}\label{lemma:compl}
If $\subg{\mathcal{O}}{S}$ is a consistent subgraph of $\mathcal{O}$, $\subg{\reverse{\mathcal{O}}}{-S}$ is a consistent subgraph of $\reverse{\mathcal{O}}$.
\end{lemma}

\begin{proof}
We prove this Lemma by contradiction. For $\forall u \in V-S$ and $\forall v \in \anceR{u}$,\footnote{We use $\anceR{u}$ and $\descR{u}$ for ancestors and descendants of $u$ in $\reverse{\mathcal{O}}$; $\anceR{u} = \desc{u}$ and $\descR{u} = \ance{u}$.} if $v \notin V-S$, then $u \in \ance{v} \subseteq S$ due to the consistency of $\subg{\mathcal{O}}{S}$. This contradicts $u \in V-S$. Therefore, the assumption $v \notin V-S$ is false and we have $\forall u \in V-S$, $\forall v \in \anceR{u} \subseteq V-S$. That is, $\subg{\reverse{\mathcal{O}}}{-S}$ is consistent.
\end{proof}

This Lemma demonstrates that all complementary white vertices in Figure~\ref{fig:example} form consistent subgraphs in the reverse graph.

\begin{theorem}\label{thm:reverse}
Given an ontology $\mathcal{O}$, $\cdag(\mathcal{O}) = \cdag(\reverse{\mathcal{O}})$.
\end{theorem}

\begin{proof}
Given Lemma~\ref{lemma:compl}, we see that the mapping $f(\subg{\mathcal{O}}{S}) = \subg{\reverse{\mathcal{O}}}{-S}$ is a \emph{bijection} between the two sets of consistent subgraphs. Therefore, the two sets are of equal cardinality.
\end{proof}

Theorem~\ref{thm:reverse} permits graph reversal at any point during the algorithm depending on which of the graphs is more likely to terminate first. For example, we can always choose the one with fewer multi-parent vertices so as to greedily reduce the upper bound of recursive calls. It is worth noting that all the leaves become roots in the reverse graph. Therefore, in the final algorithm that incorporates both pruning and reversing modules, we generalize the algorithm to allow for $\varphi > 1$ on roots (branching vertices in the reverse sense) in order to ensure compatibility.

Having $\varphi(r)>1$ on a root indicates that all the ancestors of $r$ have been pruned out. For trees (after pruning), we have $\subg{\mathcal{O}}{-\descPlus{r}} = \subg{\mathcal{O}}{\ance{r}}$. With Lemma~\ref{lemma:withoutD} and Theorem~\ref{thm:reverse}, we have
\begin{align*}
    \cdag(\mathcal{O}\vert \neg r) = \cdag(\subg{\mathcal{O}}{-\descPlus{r}}) = \cdag(\subg{\mathcal{O}}{\ance{r}}) = \cdag(\subg{\reverse{\mathcal{O}}}{\mathcal{D}^R(r)}) = \varphi(r).
\end{align*}
On the other hand, for any consistent subgraph $S$ containing $r$, $S-\ancePlus{r}$ induces a consistent subgraph in $\subg{\mathcal{O}}{\desc{r}}$ and vice versa; thus, 
\begin{align*}
    \cdag(\mathcal{O}\vert r) = \prod_{u\in \mathcal{C}(r)}\ctree(\mathcal{T}_u).
\end{align*}
Hence, these two subtotals sum to be the total count and Equation~\ref{eqn:tree_recursion_gen} remains unchanged. However, if a root $r$ with $\varphi(r)>1$ is selected to be the pivot, we have the following equation according to Theorem~\ref{thm:reverse} and Equation~\ref{eqn:decompose_gen},
\begin{align*}
    \cdag(\mathcal{O}) = \cdag(\reverse{\mathcal{O}}) &= \cdag(\subg{\reverse{\mathcal{O}}}{-\mathcal{D}^{R+}(r)}) + \varphi(r)\cdot\cdag(\subg{\reverse{\mathcal{O}}}{-\mathcal{A}^{R+}(r)}) \\
    &=
    \cdag(\subg{\mathcal{O}}{-\ancePlus{r}}) + \varphi(r)\cdot\cdag(\subg{\mathcal{O}}{-\descPlus{r}}),
\end{align*}
\noindent whereas Equation~\ref{eqn:decompose_gen} remains unchanged for non-root vertices.

\subsection{Pivot Selection}
\label{sec:pivotselection}

As alluded to before, the selection of vertices used for partitioning has the potential to significantly change the computation time. It is therefore reasonable to devise a strategy for pivot selection. Besides a random selection of multi-parent vertices~(mpv's), which aims at directly converting DAGs into trees one step at a time, we also consider three other pivot heuristics. The first strategy is to pick a vertex with the maximum degree, with random selection in case of ties, because decomposing the graph according to such vertices may increase the chance of having either disconnected components or branching components. The second strategy selects the pivot so as to minimize $e-n+r$ over the two subproblems, where $e$, $n$, and $r$ are the number of edges, vertices, and roots in the two components. We refer to this quantity as ``bound'' since it is an upper bound of the number of mpv's in the graph (see Supplementary Materials for the proof). Note that it is closely related to the cyclomatic number of the graph. Finally, the third strategy simulates a unit network flow for all vertices running in the direction from leaves to the roots and selects the ``bottleneck'' vertex; i.e., the one that maximizes the ratio of the flow in the vertex and the number of its descendants (see Supplementary Materials for this pivot selection algorithm). These strategies will be empirically compared in Section~\ref{sec:results}.

\subsection{Hashing}\label{sec:hashing}
It can occur during the recursive procedure that certain subgraphs require repeated enumeration. In Figure~\ref{fig:pivot}, for example, the subgraph $h$-$i$-$j$ is present in both subproblems shown in Figures~\ref{fig:pivot}b-c. Computing the count for this subgraph would emerge in the Figure~\ref{fig:pivot}b subproblem if the ensuing decomposition were based on vertex $d$, although it would not emerge if the partitioning were based on vertex $j$. Interestingly, the subgraph $k$-$l$ would be counted twice in the Figure~\ref{fig:pivot}b subproblem; i.e., when both $\ancePlus{d}$ and $\descPlus{d}$ are removed, and it would then appear one more time in the Figure~\ref{fig:pivot}c subproblem.

To avoid repeated enumeration, whenever a solution to a subproblem is obtained, the count for this subproblem is stored. Then, during the recursive calls, we first check if the result is already available before further calculation. To hash a result, we use the sorted IDs of all vertices in the subgraph as a key. Obviously, this key is unique because it corresponds to a vertex-induced subgraph of $\mathcal{O}$. For the pruned subgraph, we store the key of the subgraph along with the branching vertex. Whenever the ID of the branching vertex is used to generate a key, the stored key of the corresponding subgraph is appended to the vertex's ID with parentheses around it.

\begin{algorithm}[ht]\footnotesize
\SetKwInOut{Input}{Input}
\SetKwInOut{Output}{Output}
\SetKwBlock{Begin}{begin}{end}
\Input{A directed acyclic graph $\mathcal{O}$}
\Output{The number of consistent subgraphs in $\mathcal{O}$.}
\SetAlgoLined
\SetKwProg{Fn}{Function}{}{end}
\Fn{$\cdag^{\ast}(\mathcal{O})$}{
    $(count, succeed) \gets \texttt{lookup\_hash\_table}(\mathcal{O})$\;
    \If{$succeed$}{
        \Return $count$\;
    }
    \If(\tcp*[f]{check the number of multi-parent vertices}){$m(\reverse{\mathcal{O}}) < m(\mathcal{O})$}{
        $\mathcal{O} \gets$ \texttt{reverse}($\mathcal{O}$)\;
    }
    $count \gets 1$\;
    \ForEach{connected component $\mathcal{O}_{i}$}{
        \uIf{$\mathcal{O}_{i}$ is a tree}
        {
            $count_i \gets \ctree^{\ast}(\mathcal{O}_{i})$\;
        }
        \Else
        {
            $\{u_j\} \gets \texttt{branching\_vertices}(\mathcal{O}_{i})$\;
            \ForEach{$u_j$ in a reversed topological order}
            {
                $\varphi(u_j) \gets \cdag^{\ast}(\subg{\mathcal{O}_{i}}{\desc{u_j}})$\;
                $\texttt{prune}(\desc{u_j})$\;
            }
            $u \gets \texttt{pivot}(\mathcal{O}^{st}_{i})$\;
            \uIf{$u$ is a root}
            {
                $count_i \gets \cdag^{\ast}(\subg{\mathcal{O}^{st}_{i}}{-\ancePlus{u}}) + \varphi(u)\cdot\cdag^{\ast}(\subg{\mathcal{O}^{st}_{i}}{-\descPlus{u}})$\;
            }
            \Else
            {
                $count_i \gets \cdag^{\ast}(\subg{\mathcal{O}^{st}_{i}}{-\descPlus{u}}) + \varphi(u)\cdot\cdag^{\ast}(\subg{\mathcal{O}^{st}_{i}}{-\ancePlus{u}})$\;
            }
        }
        $count \gets count \cdot count_i$\;
    }
    $\texttt{insert\_hash\_table}(\mathcal{O}, count)$\;
    \Return $count$\;
}
\caption{The advanced version of Algorithm~\ref{algo:dag} with optimization modules.}
\label{algo:fulldag}
\end{algorithm}

\section{Experiments and Results}
\label{sec:results}
We empirically evaluate the enumeration procedure from Algorithm \ref{algo:fulldag} and various practical speedups using randomly generated graphs. We then apply this algorithm to four biomedical ontologies to gain insight into the sizes of their concept annotation spaces.

\subsection{Run-time Evaluation}
\label{sec:runtime}
We generated two sets of graphs to investigate the efficacy of our algorithm. Each set contained 1000 graphs with either 25  or 100 vertices. To construct each graph the vertices were added sequentially, with the proposed in-degree in-deg($v$) of the $k$-th vertex $v$ generated according to a Poisson distribution with parameter $\lambda$. This vertex then became a child of $\min$(in-deg($v$),~$k-1$) previously generated vertices that were themselves selected uniformly randomly. The parameter $\lambda$ was selected according to the $\Gamma(2.0, 1.0)$ prior for each new graph and kept constant until the graph was completed. 

With these two sets of simulated graphs, we ran our algorithm with different modules and pivot selection strategies. In particular, we evaluate pivot selection based on (i) random selection of vertices, (ii) random selection of multi-parent vertices, (iii) the degree criterion, (iv) the bound criterion, and (v) the bottleneck criterion. For each pivoting strategy, we subsequently add the pruning component, then hashing, and finally the graph reversal. The criterion for graph reversal was the number of multi-parent vertices; i.e., a graph will be reversed at any point during the recursive process if the reversed graph contains fewer multi-parent vertices.

We report the average wall-time and average number of recursive calls over the two sets of 1000 graphs ($|V|=25$ in Table \ref{tab:25nodes}; $|V|=100$ in Table \ref{tab:100nodes}). For the smaller graphs, we also ran a brute-force algorithm that was further convenient to empirically evaluate the correctness of our algorithm. We see that simpler schemes perform better on small graphs where the number of recursive calls per graph has not exceeded a few hundreds. On the other hand, the advanced techniques show tangible benefits on the larger graphs reducing the number of recursive calls and total computation time by orders of magnitude. It is possible to envision other variations that could result in further speedups; e.g., selecting multi-parent pivots with the highest degree. These refinements, however, were beyond the scope of this paper.

\begin{table*}\small
    \centering
    \caption{Experiments with simulated graphs with $|V|=25$ vertices. Each field in the table summarizes the per-graph wall-time over a set of 1000 graphs as well as the per-graph number of recursive calls, except for the brute-force method. The columns represent pivot selection strategies: (i) random, (ii) random multi-parent vertex (mpv), (iii) minimum bound, (iv) maximum degree, and (v) bottleneck. The rows represent successive additions of practical modules for speedups: (\protect\drawcircle{white}) basic approach from Algorithm \ref{algo:dag}, (\protect\drawcircle{myred}) pruning, (\protect\drawcircle{myred}\;\protect\drawcircle{myyellow}) pruning and hashing, (\protect\drawcircle{myred}\;\protect\drawcircle{myyellow}\;\protect\drawcircle{myblue}) pruning, hashing, and graph reversal.}
    \begin{tabular}{
        |c|
        l|
        r@{\hspace{1pt}}r|
        r@{\hspace{1pt}}r|
        r@{\hspace{1pt}}r|
        r@{\hspace{1pt}}r|
        r@{\hspace{1pt}}r|
        } \hline
        brute-force & module   & \multicolumn{2}{c|}{random}     & \multicolumn{2}{c|}{random mpv} & \multicolumn{2}{c|}{min.~bound} & \multicolumn{2}{c|}{max.~degree} & \multicolumn{2}{c|}{bottleneck} \\ \hline \hline
        \multirow{4}{*}{571ms} &\drawcircle{white}
                               & 22.5ms & (313) & 5.3ms & (39)   & 25.7ms & (23)  & 1.2ms & (28)    & 18.2ms & (47) \\ \cline{2-12}
                               &\drawcircle{myred}
                               & 21.1ms & (97)  & 14.3ms & (44)  & 26.4ms & (25)  & 10.2ms & (28)   & 25.3ms & (44) \\ \cline{2-12}
                               &\drawcircle{myred}\;\drawcircle{myyellow}
                               & 19.6ms & (71)  & 14.4ms & (39)  & 26.3ms & (23)  & 10.2ms & (26)   & 24.9ms & (34) \\ \cline{2-12}
                               &\drawcircle{myred}\;\drawcircle{myyellow}\;\drawcircle{myblue}
                               & 19.2ms & (67)  & 7.5ms & (28)   & 25.5ms & (23)  & 7.5ms & (23)    & 23.9ms & (31) \\ \hline
    \end{tabular}
    \label{tab:25nodes}
\end{table*}

\begin{table*}\small
    \centering
    \caption{Experiments with simulated graphs with $|V|=100$ vertices, with rows and columns identical to those in Table \ref{tab:25nodes}. The entries with an asterisk indicate that a sample of graphs was considered (instead of a full set of 1000) due to the long run-time. The brute-force algorithm was not considered as it was not feasible to compute the count for even a single graph.}
    \begin{tabular}{
        |l|
        r@{\hspace{1pt}}r|
        r@{\hspace{1pt}}r|
        r@{\hspace{1pt}}r|
        r@{\hspace{1pt}}r|
        r@{\hspace{1pt}}r|
        } \hline
        module & \multicolumn{2}{c|}{random}         & \multicolumn{2}{c|}{random mpv}       & \multicolumn{2}{c|}{min.~bound}    & \multicolumn{2}{c|}{max.~degree}   & \multicolumn{2}{c|}{bottleneck} \\ \hline \hline
        \drawcircle{white}
        & *{3,102}s & ({119,745,876})   & 5.21s & (52,954)      & 114s & (25,416)   & 9.93s & (101,342) & 122s & (526,925)  \\ \hline
        \drawcircle{myred}
        & *{241}s & ({1,727,306})       & 8.98s & (33,271)      & 3.93s & (2,802)   & 1.10s & (3,066)   & 4.28s & (12,597)  \\ \hline
        \drawcircle{myred}\;\drawcircle{myyellow}
        & *{132}s & ({387,910})        & 7.35s & (14,913)      & 3.67s & (1,111)   & 0.92s & (1,107)   & 3.08s & (2,052)   \\ \hline
        \drawcircle{myred}\;\drawcircle{myyellow}\;\drawcircle{myblue}
        & 165s & (508,521)         & 4.68s & (9,721)       & 3.22s & (1,103)   & 0.84s & (1,079)   & 2.79s & (2,133)   \\ \hline
    \end{tabular}
    \label{tab:100nodes}
\end{table*}

\subsection{Consistent Subgraphs in Biomedical Ontologies}
\label{sec:gohpocounts}

We use 02/2017 versions of Gene Ontology (GO) and Human Phenotype Ontology (HPO) as the target ontologies and compute the number of consistent subgraphs in each of them. The algorithm is applied to each of the three domains of GO \cite{Ashburner2000}: (i) molecular function ontology~(MFO; 10,789 terms) (ii) biological process ontology~(BPO; 29,575 terms) and (iii) cellular component ontology~(CCO; 4,085 terms). Together with HPO (12,167 terms), these four ontologies are widely used in annotating functional terms of gene products \cite{Radivojac2013, Jiang2016}. We further define the annotation \emph{level} for each term in the ontology to be the length of the longest path to the root. Starting from the root term, we add more specific terms level-by-level so as to understand how the potential annotation space grows with increased granularity of functional concepts.

In addition to level-wise full ontologies, we also investigate the ``used'' ontologies in which each term was retained only if at least one protein in the UniProt-GOA \cite{Huntley2015} and HPO \cite{Robinson2010} databases has been confidently assigned that term (confident annotations include all experimental evidence codes as well as ``traceable author statement'' and ``inferred by curators''). Protein function annotations were extracted from the 02/13/2017 release of the UniProt-GOA database, which contains 64,362 proteins with confident MFO annotations, 84,413 proteins with BPO annotations and 79,630 proteins with CCO annotations. HPO annotations were extracted from the 02/24/2017 release of the HPO database where 6,411 genes with confident annotations were extracted.

Figure~\ref{fig:results} shows the completed counts for both full and used level-wise ontologies. For each ontology, we additionally compute the lower bound (generally the larger of $2^{\ell}$ and $2^{r}$, where $r$ is the number of roots) and estimate the upper bound (we convert a graph into a forest by keeping only one randomly selected incoming edge for each multi-parent vertex and then call $\cforest$).

\begin{figure}
    \centering
    \includegraphics[width=.48\textwidth]{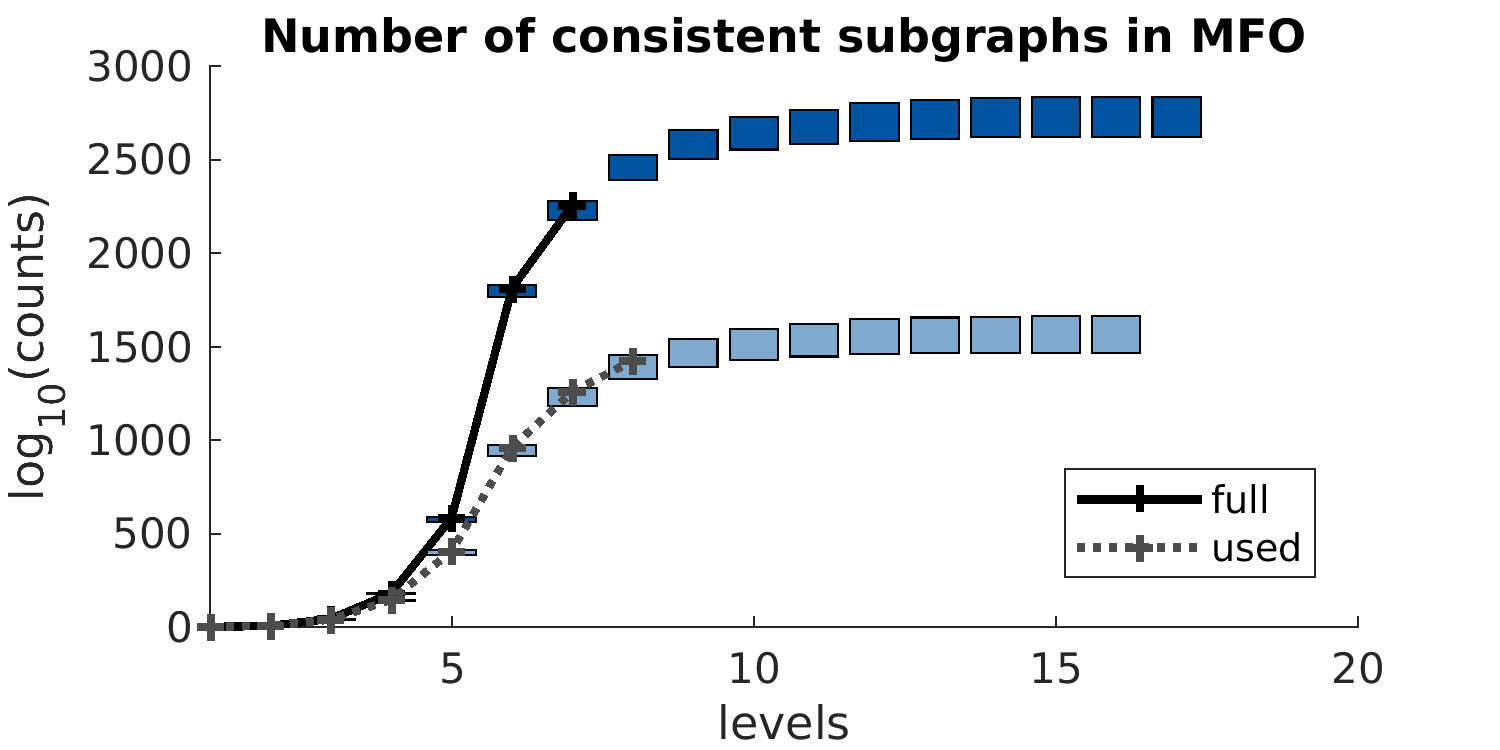}
    \includegraphics[width=.48\textwidth]{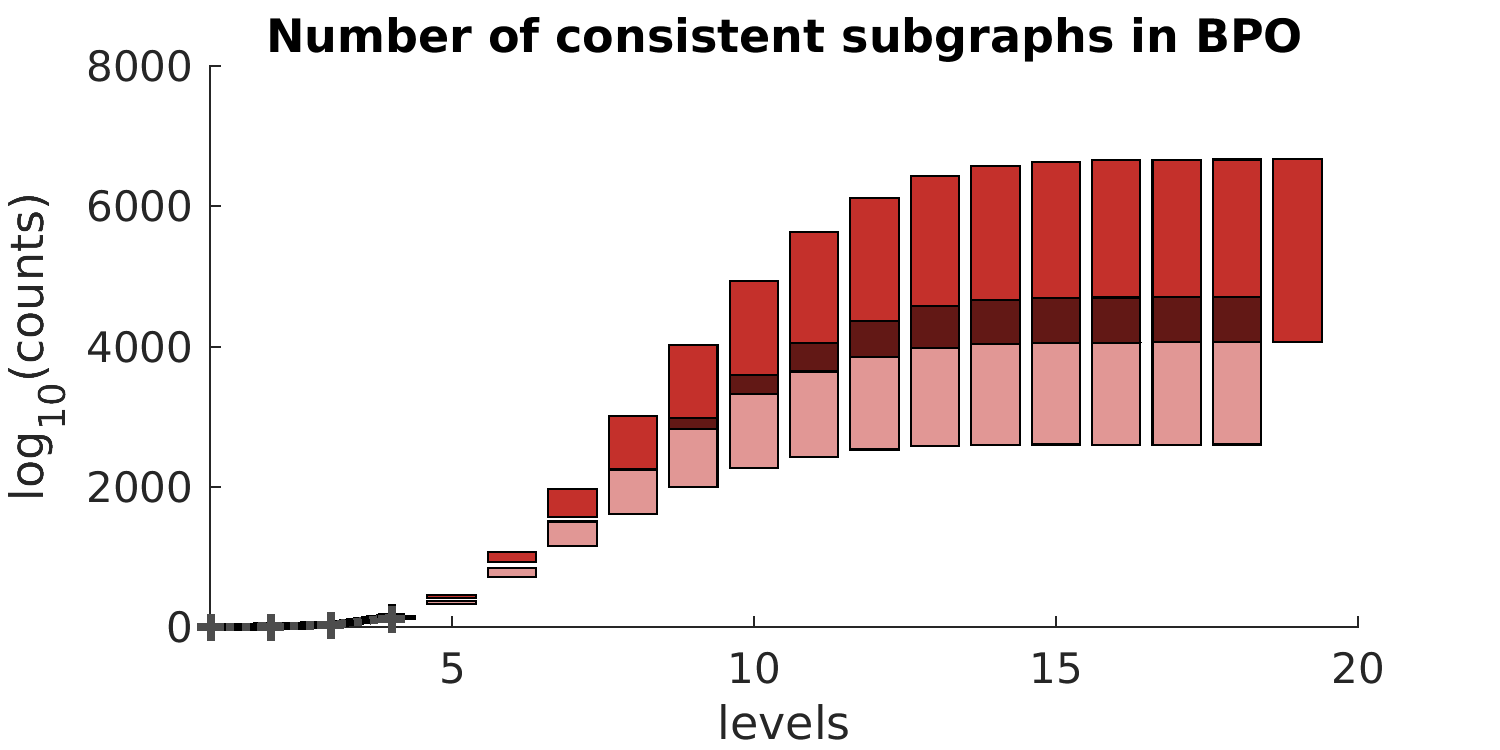} \\
    \includegraphics[width=.48\textwidth]{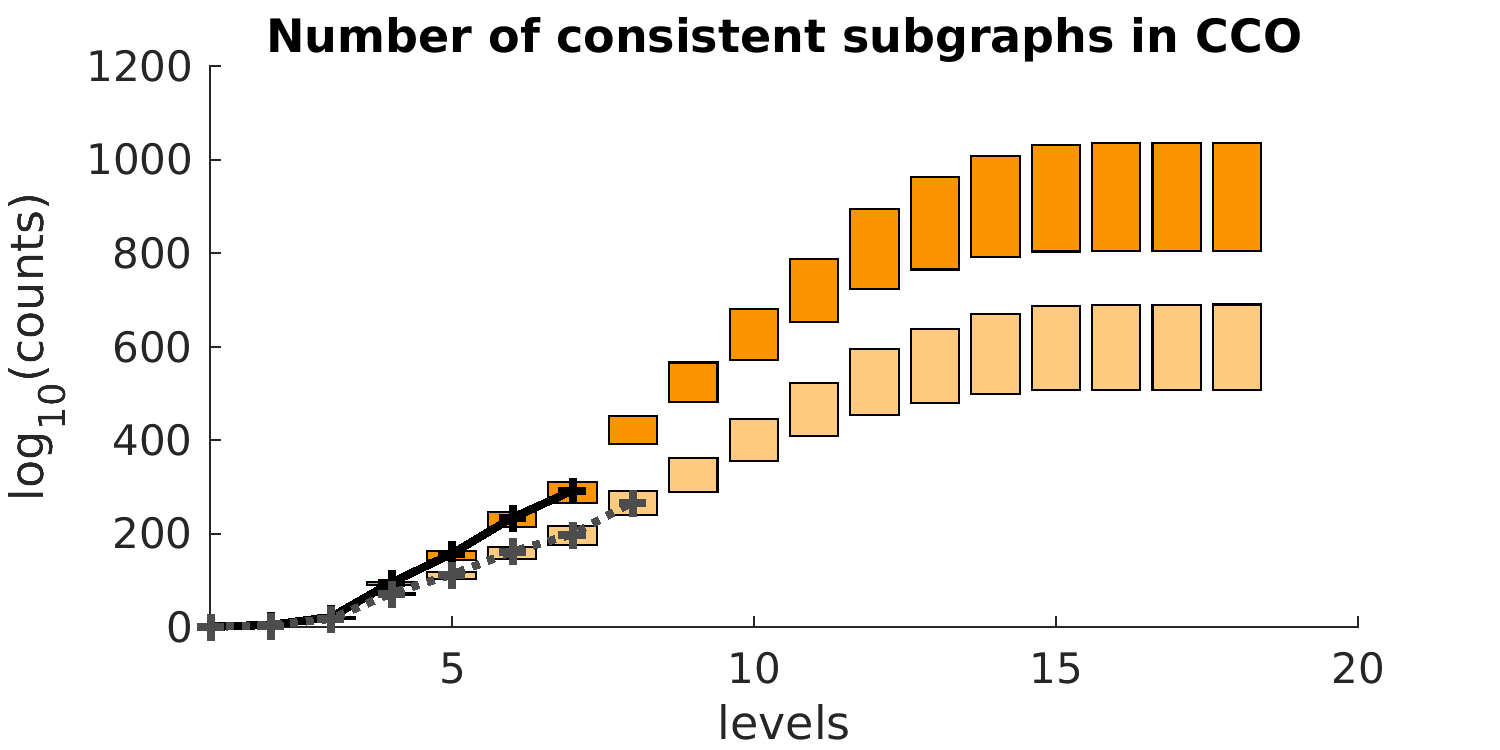}
    \includegraphics[width=.48\textwidth]{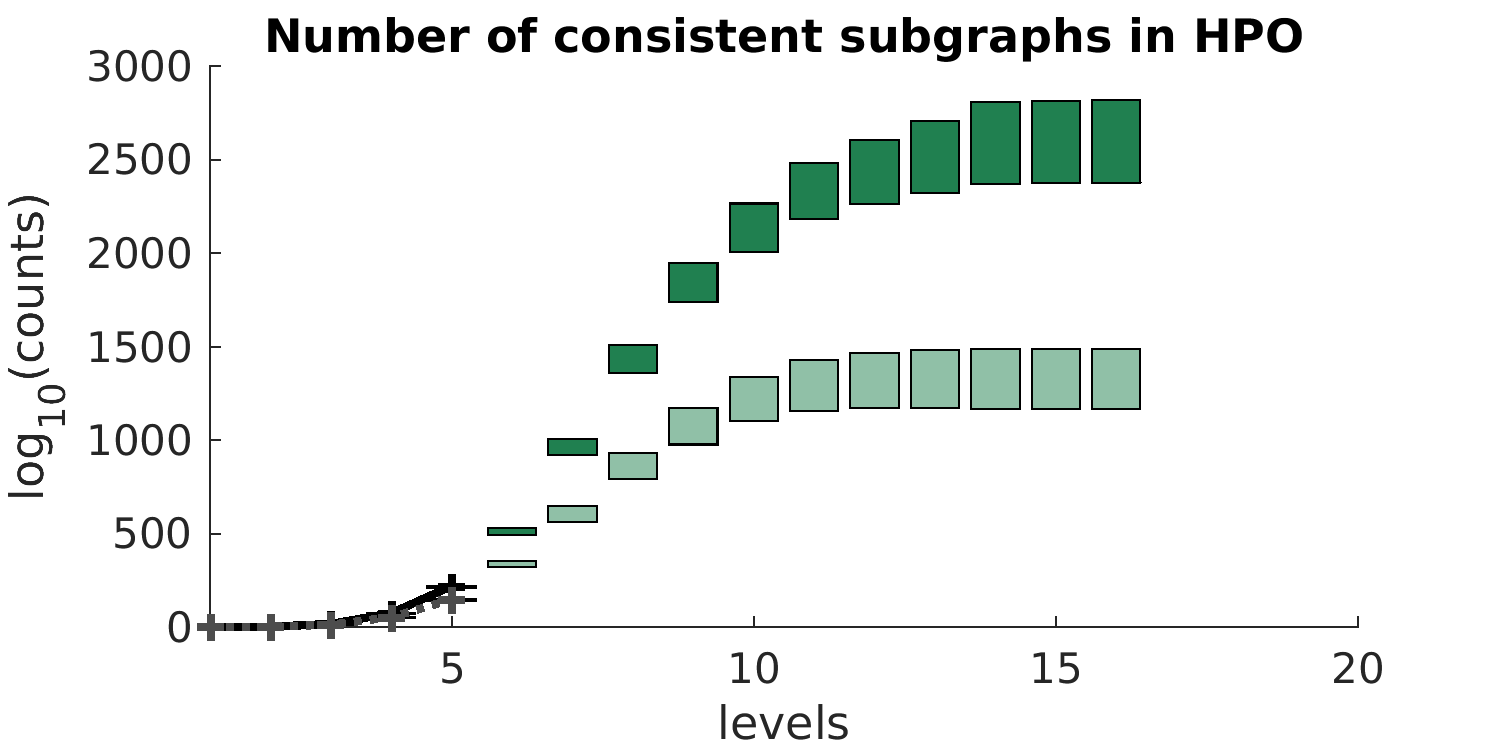}
    \caption{Number of consistent subgraphs in level-wise GO and HPO. In each panel, solid black lines mark the exact counts for ``full'' subgraphs and grey dotted lines mark the exact counts for ``used'' subgraphs. Colored bars indicate the estimated upper/lower bounds of the actual counts. The exact integer counts are available upon request.}
    \label{fig:results}
\end{figure}

Although we were not surprised by the astronomical sizes of concept annotation spaces, it was interesting to quantify them whenever feasible as well as to observe an increasing difference between lower and upper bounds (in the 1000s of orders of magnitude) with the level of the ontology. We also find it interesting that a large number of ontological terms have never been used (see Supplementary Materials). These outcomes raise questions regarding the predictability of ontological annotations as most modern algorithms are asked to provide accurate deep annotations to be deemed useful. However, annotation spaces become exceedingly large almost instantaneously with the depth of the ontology, which presents an immense computational and statistical challenge for any prediction algorithm. We therefore believe that the balance between ontology size/complexity and term granularity might become an important topic for future discussions.

\subsection{Entropy of Concept Annotation Spaces}
\label{sec:entropy}
The ability to enumerate subgraphs in relatively large ontologies presents an opportunity to contrast the space of actual ontological annotations in biological databases with the space of possible ontological annotations. To investigate this, we first computed the entropy of actual annotations at different levels in the ontology,

\begin{align*}
    H(\mathcal{O}_{lvl}) = - \sum_{i} P(\subg{\mathcal{O}_{lvl}}{S_i}) \log_2 P(\subg{\mathcal{O}_{lvl}}{S_i}),
\end{align*}

\noindent where $\mathcal{O}_{lvl}$ is the truncated ontology as in Section \ref{sec:gohpocounts}, $\subg{\mathcal{O}_{lvl}}{S_i}$ corresponds to a distinct consistent subgraph annotation observed at that level and $P(\subg{\mathcal{O}_{lvl}}{S_i})$ is the probability that a protein is assigned annotation $\subg{\mathcal{O}_{lvl}}{S_i}$. We first enumerated all observed subgraphs from the UniProt-GOA or HPO database truncated to a particular level, calculated their relative frequencies, and then plugged these relative frequencies into the entropy formula above. On the other hand, the maximum entropy was computed as $\log _{2} \cdag(\mathcal{O}_{lvl})$ by assuming equal probability for every possible consistent subgraphs.

\vspace{2mm}

\begin{figure}[!ht]
    \centering
    \noindent
    \begin{minipage}[t]{.62\textwidth}
        \centering
        \vspace{0pt}
        \includegraphics[width=\linewidth]{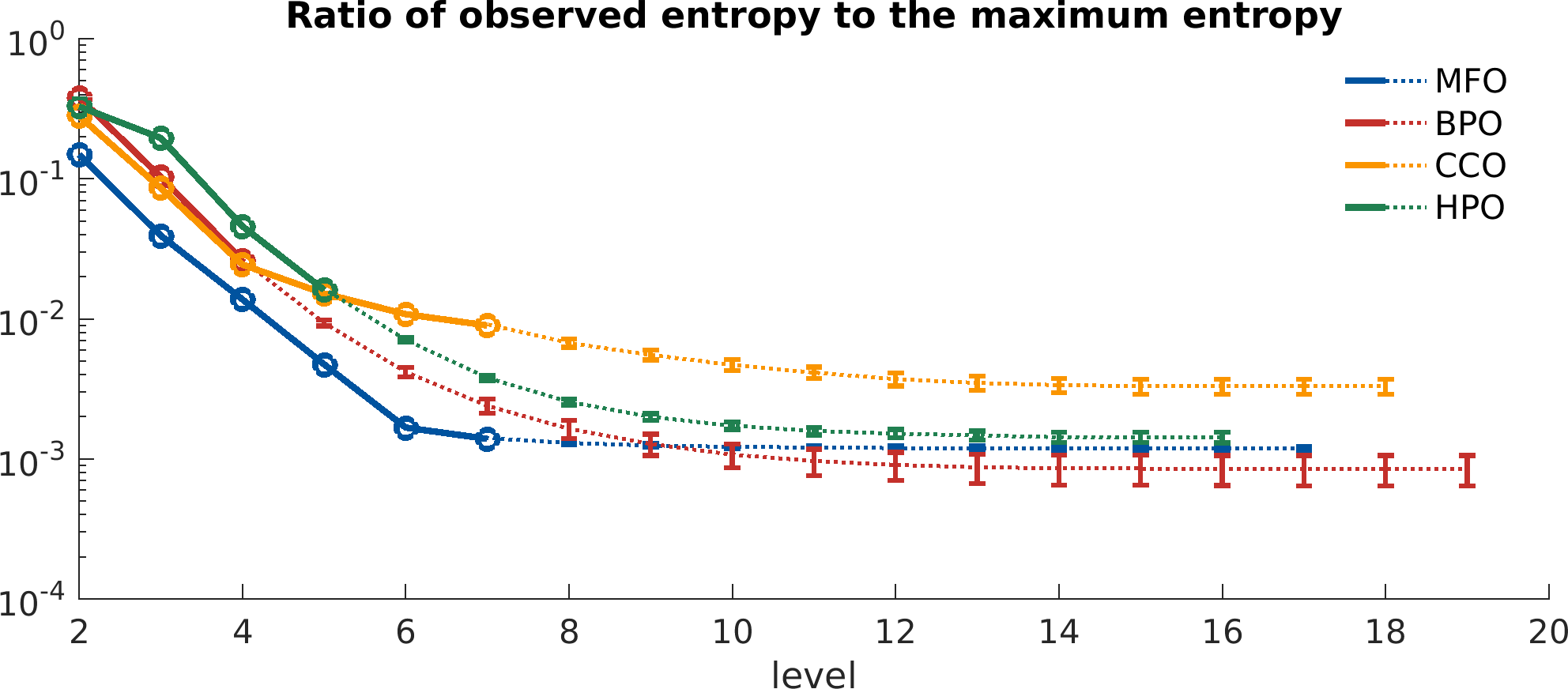}
    \end{minipage}%
    \hfill
    \begin{minipage}[t]{.36\textwidth}
        \vspace{0pt}
        \captionof{figure}{Ratio of entropies in the four ontologies. Circles with solid lines show the ratio of observed entropy to the maximum entropy. Dotted lines correspond to the estimated ratios as the average of the two ratios calculated by lower/upper bound of the counts. The error bars suggest a possible placement for the actual ratio.}
        \label{fig:entropy}
    \end{minipage}%
\end{figure}

Figure~\ref{fig:entropy} shows the ratio between the two quantities for levels greater than 1, suggesting that the world of protein functions, despite great diversity, has low entropy relative to the possible maximum. Although the currently observed functional annotations are incomplete, noisy and biased \cite{Schnoes2009, Schnoes2013, Jiang2014}, this suggests considerable departure from the uniform distribution and implies that an extensive number of possible consistent subgraphs have not been used.

\section{Related Work}
There exists a body of literature in enumerative combinatorics related to our work. One of the most relevant problems is the enumeration of directed acyclic graphs with $n$ distinct (labeled) nodes \cite{Robinson1971}. The resulting count reflects the size of the structure space of Bayesian networks with $n$ random variables and, surprisingly, also corresponds to the number of matrices in $\left\{ 0,1\right\} ^{n\times n}$ with all eigenvalues real and positive \cite{McKay2004}. The number of labeled directed acyclic graphs with $n$ nodes does not have a closed-form solution and is instead available as the A003024 sequence in the On-Line Encyclopedia of Integer Sequences (OEIS); \url{https://oeis.org/A003024}. The construction was originally proposed by Robinson \cite{Robinson1971} and was further investigated by others \cite{Stanley1973, Rodionov1992, Gessel1996}.

The results on rooted labeled trees include both the enumeration of possible number of trees and also the enumeration of subtrees for a given tree. There are $n^{n-1}$ labeled rooted trees with $n$ nodes \cite{Gross2004} that provide the integer sequence A000169 in OEIS; \url{https://oeis.org/A000169}. The expansion to forests gives $(n+1)^{n-1}$ using Cayley's formula \cite{Cayley1889}, as a single root can be added to connect a forest of unrooted labeled trees into a rooted labeled tree. The recurrence for the number of subtrees of a given tree was proposed by Ruskey \cite{Ruskey1981}; see Algorithm \ref{algo:tree}. The generalization to weighted subtrees was given by Yan and Yeh \cite{Yan2006}. Both algorithms are linear in $n$ assuming constant time addition and multiplication.

Our work also relates to the research in ontology quality assurance. These efforts typically include the analysis of irregularities and redundancy in concept descriptors and graph structure \cite{Bodenreider2003, Verspoor2009, Xing2016}. Our work, primarily the software we developed, contributes to this area by facilitating the analysis of the annotation space.

\section{Conclusions}
This work presents a practical algorithm for enumerating consistent subgraphs of directed acyclic graphs. Although this study largely addresses an intellectual challenge of efficient subgraph enumeration, it might have practical utility for the studies of annotation spaces in the biomedical sciences and beyond; e.g., in text mining \cite{Grosshans2014} and computer vision \cite{Movshovitz2015}. As ontologies are easy to grow and hard to manually interrogate, we provide a practical tool that can give insights into the complexity of concept annotation tasks \cite{Radivojac2013, Jiang2016}. As such, it may serve as a guide to ontology developers.

\section*{Acknowledgements}
This work has been supported by the National Science Foundation grant DBI-1458477 and the Indiana University Precision Health Initiative. 

\bibliographystyle{abbrv}
\bibliography{refs}

\end{document}